[1994/12/01]
\documentclass{aomart}
\pdfoutput=1

\usepackage{tikz}
\usetikzlibrary{calc,arrows,automata}
\usetikzlibrary{matrix,positioning,arrows,decorations.pathmorphing,shapes}
\usetikzlibrary{shapes,snakes}
\chardef\bslash=`\\ 

\newtheorem[{}\it]{thm}{Theorem}[section]

\newtheorem{lem}[thm]{Lemma}
\newtheorem{prop}[thm]{Proposition}

\theoremstyle{definition}

\newtheorem*[{}\it]{notation}{Notation}

\newcommand{\eval}[2][\right]{\relax
  \ifx#1\right\relax \left.\fi#2#1\rvert}

\title[Diameter Constrained Reliability: Computational Complexity in terms of the diameter and number of terminals]
{Diameter Constrained Reliability: Computational Complexity in terms of the diameter and number of terminals}
\author{Eduardo Canale, Pablo Romero}
\email{canale@fing.edu.uy; promero@fing.edu.uy} 
\address{Universidad de la Rep\'ublica\\
Montevideo, Uruguay}
\fulladdress{Instituto de Matem\'atica y Estad\'istica, Prof. Ing. Rafael Laguardia. IMERL\\
Julio Herrera y Reissig 565\\
PC, 11300}

\subject{primary}{matsc2000}{123456789}
\subject{secondary}{matsc2000}{FFFF}

\volumenumber{1}
\issuenumber{1}
\publicationyear{2014}

\begin{document}

\begin{abstract}
Let $G=(V,E)$ be a simple graph with $|V|=n$ nodes and $|E|=m$ links, a subset  
$K \subseteq V$ of \emph{terminals}, a vector $p=(p_1,\ldots,p_m) \in [0,1]^m$ and a positive integer $d$, called \emph{diameter}. 
We assume nodes are perfect but links fail stochastically and independently, with 
probabilities $q_i=1-p_i$. The \emph{diameter-constrained reliability} (DCR for short), is the probability that the 
terminals of the resulting subgraph remain connected by paths composed by $d$ links, or less. 
This number is denoted by $R_{K,G}^{d}(p)$.

The general DCR computation is inside the class of $\mathcal{N}\mathcal{P}$-Hard problems, since is subsumes the complexity that a random graph is connected.  
In this paper, the computational complexity of DCR-subproblems is discussed in terms of the number of terminal nodes $k=|K|$ and diameter $d$. 
Either when $d=1$ or when $d=2$ and $k$ is fixed, the DCR is inside the class $\mathcal{P}$ of polynomial-time problems.  
The DCR turns $\mathcal{N}\mathcal{P}$-Hard when $k \geq 2$ is a fixed input parameter and $d\geq 3$.

The case where $k=n$ and $d \geq 2$ is fixed are not studied in prior literature. 
Here, the $\mathcal{N}\mathcal{P}$-Hardness of this case is established. 
\end{abstract}

\maketitle
\tableofcontents

\section{Introduction}
The definition of DCR has been introduced by H\'ector Cancela and Louis Petingi, inspired in delay-sensitive applications over 
the Internet infrastructure~\cite{PR01}. 
Nevertheless, its applications over other fields of knowledge enriches the motivation 
of this problem in the research community~\cite{Colbourn99reliabilityissues}. 

We wish to communicate special nodes in a network, called \emph{terminals}, 
by $d$ hops or less, in a scenario where nodes are perfect but links fail stochastically and independently. 
The all-terminal case with $d=n-1$ is precisely the probability that a random graph is connected, or 
\emph{classical reliability problem} (CLR for short). Arnon Rosenthal proved that the CLR is inside the class 
of $\mathcal{N}\mathcal{P}$-Hard problems~\cite{Rosenthal}. As a corollary, the general DCR is $\mathcal{N}\mathcal{P}$-Hard as well, 
hence intractable unless $\mathcal{P}=\mathcal{N}\mathcal{P}$. 

The focus of this paper is the computational complexity of DCR subproblems, in terms of the number of terminals $k$ and 
diameter $d$. In Section~\ref{DCR} a formal definition of DCR is provided as a particular instance of a coherent stochastic 
binary system. The computational complexity of the DCR is discussed in terms of the diameter and number of terminals in Section~\ref{Prior}. 
The main contribution of this paper is included in Section~\ref{Result}. Specifically, the DCR is in the computational class of 
$\mathcal{N}\mathcal{P}$-Hard problems in the all-terminal scenario ($k=n$) with a given diameter $d \geq 2$. 
This result closes the complexity analysis of the DCR in terms of $k$ and $d$. 
Concluding remarks and open problems are summarized in Section~\ref{Conclusion}.

\section{Terminology}\label{DCR}
We are given a system with $m$ components. These components are either ``up'' or ``down'', and 
the binary state is captured by a word $x=(x_1,\ldots,x_m)$. Additionally we have a 
structure function $\phi: \{0,1\}^m \to \{0,1\}$ such that $\phi(x)=1$ if and only if the system works 
under state $x$. When the components work independently and stochastically 
with certain probabilities of operation $p=(p_1,\ldots,p_m)$, the pair $(\phi,p)$ defines a \emph{stochastic binary system}, 
or SBS for short, following the terminology of Michael Ball~\cite{Ball1986}. 
An SBS is \emph{coherent} whenever $x\leq y$ implies that $\phi(x) \leq \phi(y)$, where the partial order set $(\leq,\{0,1\}^m)$ 
is bit-wise (i.e. $x\leq y$ if and only if $x_i \leq y_i$ for all $i\in \{1,\ldots,m\}$). If $\{X_i\}_{i=1,\ldots,m}$ 
is a set of independent binary random variables with $P(X_i=1)=p_i$ and $X=(X_1,\ldots,X_m)$, then 
$r=E(\phi(X)) = P(\phi(X)=1)$ is the \emph{reliability} of the SBS.

Now, consider a simple graph $G=(V,E)$, a subset $K \subseteq V$ and a positive integer $d$. 
A subgraph $G_x=(V,E_x)$ is \emph{$d$-$K$-connected} if $d_x(u,v)\leq d, \forall \{u,v\} \subseteq K$, 
where $d_x(u,v)$ is the distance between nodes $u$ and $v$ in the graph $G_x$. 
Let us choose an arbitrary order of the edge-set $E=\{e_1,\ldots,e_m\}$, $e_i\leq e_{i+1}$. 
For each subgraph $G_x=(V,E_x)$ with $E_x \subseteq E$, we identify a binary word $x \in \{0,1\}^m$, 
where $x_i=1$ if and only if $e_i \in E_x$; this is clearly a bijection. 
Therefore, we define the structure $\phi:\{0,1\}^m \to \{0,1\}$ such that $\phi(x)=1$ if and only if the graph $G_x$ 
is $d$-$K$-connected. If we assume nodes are perfect but links fail stochastically and independently ruled 
by the vector $p=(p_1,\ldots,p_m)$, the pair $(\phi,p)$ is a coherent SBS. Its reliability, 
denoted by $R_{K,G}^{d}(p)$, is called \emph{diameter constrained reliability}, or DCR for short. 
A particular case is $R_{K,G}^{m-1}(p)$, called \emph{classical reliability}, or CLR for short.

In all coherent SBS, a \emph{pathset} is a state $x$ such that $\phi(x)=1$. 
A \emph{minpath} is a state $x$ such that $\phi(x)=1$ but $\phi(y)=0$ for all $y<x$ (i.e. a minimal pathset). 
A \emph{cutset} is a state $x$ such that $\phi(x)=0$, while a \emph{mincut} is a state $x$ such that $\phi(x)=0$ 
but $\phi(y)=1$ if $y>x$ (i.e. a minimal cutset).  

Recall that a \emph{vertex cover} in a graph $G=(V,E)$ is a subset $V^{\prime} \subseteq V$ such that $V^{\prime}$ meets all links in $E$. 
The graph $G$ is bipartite if there exists a bipartition $V = V_1 \cup V_2$ such that $E \subseteq V_1 \times V_2$.

\section{Complexity}\label{Prior}
The class $\mathcal{N}\mathcal{P}$ is the set of problems polynomially solvable by 
a non-deterministic Turing machine~\cite{Garey:1979:CIG:578533}. 
A problem is $\mathcal{N}\mathcal{P}$-Hard if it is at least as hard as every problem in the set $\mathcal{N}\mathcal{P}$ 
(formally, if every problem in $\mathcal{N}\mathcal{P}$ has a polynomial reduction to the former). 
It is widely believed that $\mathcal{N}\mathcal{P}$-Hard problems are intractable (i.e. there is no polynomial-time algorithm to solve them). 
An $\mathcal{N}\mathcal{P}$-Hard problem is $\mathcal{N}\mathcal{P}$-Complete if it is inside the class $\mathcal{N}\mathcal{P}$. 
Stephen Cook proved that the joint satisfiability of an input set of clauses in disjunctive form is an $\mathcal{N}\mathcal{P}$-Complete 
decision problem; in fact, the first known problem of this class~\cite{Cook1971}. 
In this way, he provided a systematic procedure to prove that a certain problem 
is $\mathcal{N}\mathcal{P}$-Complete. Specifically, it suffices to prove that the problem is inside the class $\mathcal{N}\mathcal{P}$, 
and that it is at least as hard as an $\mathcal{N}\mathcal{P}$-Complete problem. 
Richard Karp followed this hint, and presented the first 21 combinatorial problems inside this class~\cite{Karp1972}. 
Leslie Valiant defines the class \#$\mathcal{P}$ of counting problems, such that testing whether an element 
should be counted or not can be accomplished in polynomial time~\cite{Valiant1979}. 
A problem is~\#$\mathcal{P}$-Complete if it is in the set \#-$\mathcal{P}$ and it is at least as hard as any problem of that class.  

Recognition and counting minimum cardinality mincuts/minpaths are at least as hard as computing the reliability of a coherent 
SBS~\cite{Ball1986}. Arnon Rosenthal proved the CLR is $\mathcal{N}\mathcal{P}$-Hard~\cite{Rosenthal},  
showing that the minimum cardinality mincut recognition is precisely Steiner-Tree problem, included in Richard Karp's list. 
The CLR for both two-terminal and all-terminal cases are still $\mathcal{N}\mathcal{P}$-Hard, as Michael Ball and J. Scott Provan proved 
by reduction to counting minimum cardinality $s-t$ cuts~\cite{provan83}. 
As a consequence, the general DCR is $\mathcal{N}\mathcal{P}$-Hard as well. Later effort has been focused to particular cases of the DCR, 
in terms of the number of terminals $k=|K|$ and diameter $d$.

When $d=1$ all terminals must have a direct link, $R_{K,G}^{1}= \prod_{\{u,v\} \subseteq K} p(uv)$, 
where $p(uv)$ denotes the probability of operation of link $\{u,v\} \in E$, and $p(uv)=0$ if $\{u,v\} \notin E$. 
The problem is still simple when $k=d=2$. In fact, $R_{\{u,v\},G}^{2} =1- (1-p(uv))\prod_{w \in V-\{u,v\}}(1-p(uw)p(wv))$. 
H\'ector Cancela and Louis Petingi rigorously proved that the DCR is $\mathcal{N}\mathcal{P}$-Hard when $d\geq 3$ and $k\geq 2$ 
is a fixed input parameter, in strong contrast with the case $d=k=2$~\cite{CP2004}. 
Its proof is the main source of inspiration of this paper, and will be revisited in Section~\ref{Result}. 
The literature offers at least two proofs that the DCR has a polynomial-time 
algorithm when $d=2$ and $k$ is a fixed input parameter~\cite{sartor:thesis,SartorITOR2013a}. 
Pablo Sartor et. al. present a recursive proof~\cite{sartor:thesis}, while Eduardo Canale et. al. present 
an explicit expression for $R_{K,G}^{2}$ that is computed in a polynomial time of elementary operations~\cite{SartorITOR2013a}. 
Figure~\ref{DCR-complex} summarizes the known results for the computational complexity of the DCR in terms of $d$ and $k$.

\begin{figure}[htb]\centering{
\begin{tikzpicture} [scale=1.4,nodop/.style={inner sep=0pt
}]

\path[draw] (0,1) -- (0,5) -- (5,5) -- (5,1);
\draw[dotted] (-0.3,2) -- (7,2);
\draw[dotted] (-0.3,4) -- (7,4);
\draw[dotted] (1.5,5.3) -- (1.5,4);

\node at (2.5, 5.7) {$k$ (fixed)};
\node at (-0.8, 3) {$d$};
\path[->][dotted] (-0.8, 3.2) edge (-0.8, 4);
\path[->][dotted] (-0.8, 2.8) edge (-0.8, 2);

\node at (0.75,5.3) {$2$};
\node at (1.9, 5.3) {$3\ldots$};
\node at (-0.4, 4.5) {$2$};
\node at (-0.4, 3.8) {$3$};
\node at (-0.4, 3) {.}; \node at (-0.4, 3.1) {.}; \node at (-0.4, 2.9) {.}; 
\node at (-0.4, 2.2) {$n-2$};
\node at (-0.4, 1.8) {$n-1$};
\node at (-0.4, 1.3) {.}; \node at (-0.4, 1.4) {.}; \node at (-0.4, 1.2) {.}; 

\node at (0.75, 4.5) {$O(n)$~\cite{CP2004}};
\node at (3.25, 4.5) {$O(n)$~\cite{SartorITOR2013a}};
\node at (2.5, 3) {$\mathcal{N}\mathcal{P}$-Hard~\cite{CP2004}};
\node at (2.5, 1.5) {$\mathcal{N}\mathcal{P}$-Hard~\cite{Rosenthal}};

\node at (6, 1.5) {$\mathcal{N}\mathcal{P}$-Hard~\cite{provan83}};

\path[draw] (5,5) -- (7,5) -- (7,1);
\path[draw] (5,5) -- (5,6);
\node at (6, 5.7) {$k=n$ or free};
\node at (6, 4.5) {Unknown};
\node at (6, 2.50) {Unknown};

\end{tikzpicture}} \caption{DCR Complexity in terms of the diameter $d$ and number of terminals $k=|K|$}\label{DCR-complex}
\end{figure}

\clearpage

\section{Main theorem} \label{Result}
\emph{The DCR is inside the class of $\mathcal{N}\mathcal{P}$-Hard problems in the all-terminal case with diameter $d\geq 2$}. 
We first prove the result when $d\geq 3$, and separately establish the case $d=2$. 

The main source of inspiration for the first result is the article authored by H\'ector Cancela and Louis Petingi~\cite{CP2004}, 
where they proved that the DCR is $\mathcal{N}\mathcal{P}$-Hard when 
$d\geq 3$ and $k\geq 2$ is a fixed input parameter. There, the authors prove first that the result holds for $k=2$, and 
they further generalize the result for fixed $k\geq 2$. For our purpose it will suffice to revisit the first part. 
Before, we state a technical result first proved by Michael Ball and Scott Provan~\cite{BallProvan1983b}. 
\begin{lem}~\cite{BallProvan1983b}
Counting the number of vertex covers of a bipartite graph is \#$\mathcal{P}$-Complete.
\end{lem}
 
\begin{prop}\label{Step-1}~\cite{CP2004}
The DCR is $\mathcal{N}\mathcal{P}$-Hard when $k=2$ and $d\geq 3$. 
\end{prop}
\begin{proof}
Let $d^{\prime}=d-3 \geq 0$ and $P=(V(P),E(P))$ a simple path with node set $V(P)=\{s,s_1,\ldots,s_{d^{\prime}}\}$ 
and edge set $E(P)=\{\{s,s_1\},\{s_1,s_2\},\ldots, \{s_{d^{\prime}-1},s_{d^{\prime}}\}\}$. 
For each bipartite graph $G=(V,E)$ with $V=A \cup B$ and $E\subseteq A\times B$ 
we build the following auxiliary network:
\begin{equation}
G^{\prime} = \{(A \cup B \cup V(P) \cup \{t\}, E \cup E(P) \cup I\},
\end{equation}
where $I=\{ \{s_{d^{\prime}},a\}, a\in A\} \cup \{ \{b,t \}, b\in B\}$, and 
all links of $G^{\prime}$ are perfect but links in $I$, which fail independently with identical probabilities $p=1/2$. 
Consider the terminal set $K=\{s,t\}$. The auxiliary graph $G^{\prime}$ is illustrated in Figure~\ref{fig:covercut}. 
The reduction from the bipartite graph to the two-terminal instance is polynomial.

\begin{figure}[h!]\centering{
\begin{tikzpicture}[scale=1]
\tikzstyle{every node}=[draw,shape=circle]
\path 
(-3,0) node (p0) {$s$}
(-2,0) node (p1) {$s_1$}
(-1,0) node (p2) {$s_2$}
 (0,0) node (p3) {$s_3$}

(1,-1) node (p4) {$a_1$}
(1, 0) node (p5) {$a_2$}
(1, 1) node (p6) {$a_3$}

(3,-1) node (p7) {$b_1$}
(3, 0) node (p8) {$b_2$}
(3, 1) node (p9) {$b_3$}

(5,0) node (p10) {$t$};

\draw 
(p0) -- (p1)
(p1) -- (p2)
(p2) -- (p3)

(p4) -- (p7)
(p4) -- (p9)

(p5) -- (p7)
(p5) -- (p8)

(p6) -- (p8)
(p6) -- (p9)

(p3) -- (p4)
(p3) -- (p5)
(p3) -- (p6)
(p7) -- (p10)
(p8) -- (p10)
(p9) -- (p10);
\end{tikzpicture}
} \caption{Example of auxiliary graph $G^{\prime \prime}$ with terminal set $\{s,t\}$ and $d=6$, 
for the particular bipartite instance $C_6$.}\label{fig:covercut}
\end{figure}
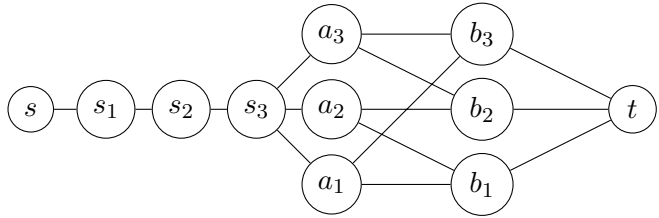
A set cover $A^{\prime} \cup B^{\prime} \subset A \cup B$ induces a cutset 
$I^{\prime} = \{ \{s_{d^{\prime}},a\}, a\in A^{\prime}\} \cup \{ \{b,t \}, b\in B^{\prime}\}$ 
(i.e. if all links in $I^{\prime}$ fail, the nodes $\{s,t\}$ are not connected). Reciprocally, that cutset determines 
a set cover. Therefore, the number of cutsets $|\mathcal{C}|$ is precisely the number of vertex covers of the bipartite graph $|\mathcal{B}|$. 
Moreover:
\begin{equation*}
|\mathcal{B}| = 2^{|A|+|B|} (1- R_{\{s,t\},G^{\prime}}^{d}(1/2)).
\end{equation*}  
Thus, the DCR for the two-terminal case is at least as hard as counting vertex covers of bipartite graphs.
\end{proof}

The result for $d \geq 3$ is perhaps a direct Corollary of Proposition~\ref{Step-1}:
\begin{thm}\label{romero}
The DCR is $\mathcal{N}\mathcal{P}$-Hard when $k=n$ and $d\geq 3$.  
\end{thm}
\begin{proof}
Extend the auxiliary graph $G^{\prime}=(V^{\prime},E^{\prime})$ to $G^{\prime \prime}=(V^{\prime \prime },E^{\prime \prime})$, 
where $V^{\prime \prime}=V^{\prime}$ and 
$E^{\prime \prime} = E^{\prime} \cup \{ \{a,a^{\prime}\}, a\neq a^{\prime}, a, a^{\prime}\in A\} 
\cup \{ \{b,b^{\prime}\}, b\neq b^{\prime},  b, b^{\prime} \in B\}$. 
In words, just add links in order to connect all nodes from $A$, and all nodes from $B$. We keep the same probabilities 
of operation that in $G^{\prime}$, and the new links are perfect. 

Consider now the all-terminal case $K=V^{\prime \prime}$ for $G^{\prime \prime}$, and given diameter $d \geq 3$. 
The key is to observe that the cutsets in the all-terminal scenario 
for $G^{\prime \prime}$ are precisely the $s-t$ cutsets in $G^{\prime}$, and they have the same probability. 
Indeed, each pair of terminals from the set $A$ are directly connected by perfect links; the same holds in $B$. 
The distance between $s$ and $s_{d^{\prime}}$ is $d^{\prime}=d-3<d$, so these nodes (and all the intermediate ones) 
respect the diameter constraint. Finally, if there were an $s-t$ path (i.e. a path from $s$ to $t$), 
the diameter of $G^{\prime \prime}$ would be exactly $d$. Therefore, $R_{\{s,t\},G^{\prime}}^{d} = R_{V^{\prime \prime},G^{\prime \prime}}^{d}$, 
and again: 
\begin{align*}
|\mathcal{B}| &= 2^{|A|+|B|} (1- R_{\{s,t\},G^{\prime}}^{d}(1/2))\\
              &= 2^{|A|+|B|} (1- R_{V^{\prime \prime},G^{\prime \prime}}^{d}(1/2)).
\end{align*}  
Thus, the DCR for the all-terminal case is at least as hard as counting vertex covers of bipartite graphs.
\end{proof}

\begin{thm}\label{canale}
The DCR is $\mathcal{N}\mathcal{P}$-Hard when $k=n$ and $d = 2$.  
\end{thm}
\begin{proof}
Given a graph $G=(V,E)$, we consider the graph 
$G^{\prime} = (V \cup \{a,b\}, E \cup \{ \{x,a\}, \{x,b\}, \forall x \in V\})$. 
By its definition, $G^{\prime}$ has diameter $d=2$. All links are perfect, except the ones incident to $a$, with $p(ax)=1/2$. 
Consider the DCR for $G^{\prime}$. We will show that the number of minimum cardinality pathsets in $G^{\prime}$ is precisely 
the number of vertex covers in $G^{\prime}$. 
Since counting minimum cardinality pathsets is at least as hard as computing the reliability of a coherent 
SBS~\cite{Ball1986}, the result will follow.

A minimum cardinality pathset in $G^{\prime}$ contains all perfect links, and some links $\{a,x_1\},\ldots,\{a,x_r\}$ 
for certain nodes $x_i \in V$. Since $H$ is a minimum cardinality pathset, 
the graph $G_H = (V,H)$ has diameter $2$, but the diameter is increased under any link deletion. 
Let $N_a = \{x: \{a,x\} \in H\}$ the set of neighbor vertices for the terminal node $a$. 

The key is to observe that \emph{vertex $a$ reaches every node in two steps if and only if $N_a$ is a vertex cover}. 
Indeed, suppose $a$ reaches every node in two steps. Then, for any $x \in V \setminus N_a$ there exists 
a path $xya$, so $y \in N_a$ and thus $N_a$ is a vertex cover. 
Conversely, if $N_a$ covers $V$, let $x \in V$. Then, either $x \in N_a$ and $x$ is adjacent with $a$, or 
$x \in V\setminus N_a$ and there exists $y \in N_a \cap N_x$, so  $xya$ is a path of two hops between $x$ and $a$.

The minimality of $N_a$ as a cover follows from the minimality of $H$ as a pathset.
\end{proof}

Theorems~\ref{romero}~and~\ref{canale} jointly close the complexity analysis for the DCR problem. 
The whole picture of DCR complexity is provided in Figure~\ref{DCR-complex}, 
which closes the complexity analysis for all independent pairs $(k,d)$.

\begin{figure}[htb]\centering{
\begin{tikzpicture} [scale=1.4,nodop/.style={inner sep=0pt
}]

\path[draw] (0,1) -- (0,5) -- (5,5) -- (5,1);
\draw[dotted] (-0.3,2) -- (7,2);
\draw[dotted] (-0.3,4) -- (7,4);
\draw[dotted] (1.5,5.3) -- (1.5,4);

\node at (2.5, 5.7) {$k$ (fixed)};
\node at (-0.8, 3) {$d$};
\path[->][dotted] (-0.8, 3.2) edge (-0.8, 4);
\path[->][dotted] (-0.8, 2.8) edge (-0.8, 2);

\node at (0.75,5.3) {$2$};
\node at (1.9, 5.3) {$3\ldots$};
\node at (-0.4, 4.5) {$2$};
\node at (-0.4, 3.8) {$3$};
\node at (-0.4, 3) {.}; \node at (-0.4, 3.1) {.}; \node at (-0.4, 2.9) {.}; 
\node at (-0.4, 2.2) {$n-2$};
\node at (-0.4, 1.8) {$n-1$};
\node at (-0.4, 1.3) {.}; \node at (-0.4, 1.4) {.}; \node at (-0.4, 1.2) {.}; 

\node at (0.75, 4.5) {$O(n)$~\cite{CP2004}};
\node at (3.25, 4.5) {$O(n)$~\cite{SartorITOR2013a}};
\node at (2.5, 3) {$\mathcal{N}\mathcal{P}$-Hard~\cite{CP2004}};
\node at (2.5, 1.5) {$\mathcal{N}\mathcal{P}$-Hard~\cite{Rosenthal}};

\node at (6, 1.5) {$\mathcal{N}\mathcal{P}$-Hard~\cite{provan83}};

\path[draw] (5,5) -- (7,5) -- (7,1);
\path[draw] (5,5) -- (5,6);
\node at (6, 5.7) {$k=n$ or free};
\node at (6, 4.5) {$\mathcal{N}\mathcal{P}$-Hard};
\node at (6, 2.50) {$\mathcal{N}\mathcal{P}$-Hard};

\end{tikzpicture}} \caption{DCR Complexity in terms of the diameter $d$ and number of terminals $k=|K|$}\label{DCR-complex}
\end{figure}

\section{Concluding Remarks} \label{Conclusion}
The reliability evaluation of a particular stochastic binary system has been discussed, 
called diameter constrained reliability (DCR). When the number of terminals $k$ or diameter $d$ 
are free, the DCR is $\mathcal{N}\mathcal{P}$-Hard, since it subsumes the classical reliability problem. 
The cases $d=1$ or $d=2$ and $k$ fixed belong to the set $\mathcal{P}$ of polynomially solvable problems. 
In contrast, the DCR turns $\mathcal{N}\mathcal{P}$-Hard when $k\geq 2$ is fixed and $d\geq 3$. 
In this paper we proved that the DCR is $\mathcal{N}\mathcal{P}$-Hard for the remaining cases (i.e. where $k=n$ and $d\geq 3$). 
As a corollary, the result holds when $d \geq 3$ and $k$ is an free parameter as well. 
The DCR remains $\mathcal{N}\mathcal{P}$-Hard for all but special cases of $k$ and $d$. 

A polytime-closed formula for the DCR has been provided in prior literature only for particular families of graphs, 
such as paths, cycles, special cases of bipartite and complete graphs, spanish fans and ladders~\cite{sartor:thesis}. 
So far, algorithmic design is focused o Monte Carlo methods and Interpolation theory~\cite{rndm13}. 
Future work is required to design approximation algorithms for the DCR, and generalize the problem to 
dependent link failures.

It is worth to notice that when all components fail independently with identical probability $p=1/2$ all graphs 
occur with the same probability. Counting the number of partial graphs with diameter $d=2$ is 
thus the DCR evaluation taking $p=1/2$. This counting problem is still open, and 
remains in the heart of graph theory.

%

\bibliography{dcr}

\providecommand{\etalchar}[1]{$^{#1}$}
\providecommand{\bysame}{\leavevmode\hbox to3em{\hrulefill}\thinspace}
\providecommand{\noopsort}[1]{}
\providecommand{\mr}[1]{\href{http://www.ams.org/mathscinet-getitem?mr=#1}{MR~%
#1}}
\providecommand{\zbl}[1]{\href{http://www.zentralblatt-math.org/zmath/en/searc%
h/?q=an:#1}{Zbl~#1}}
\providecommand{\jfm}[1]{\href{http://www.emis.de/cgi-bin/JFM-item?#1}{JFM~#1}}
\providecommand{\arxiv}[1]{\href{http://www.arxiv.org/abs/#1}{arXiv~#1}}
\providecommand{\doi}[1]{\url{http://dx.doi.org/#1}}
\providecommand{\MR}{\relax\ifhmode\unskip\space\fi MR }
\providecommand{\MRhref}[2]{%
  \href{http://www.ams.org/mathscinet-getitem?mr=#1}{#2}
}
\providecommand{\href}[2]{#2}
\begin{thebibliography}{CCR{\etalchar{+}}13}

\bibitem[Bal86]{Ball1986}
\bgroup\scshape{}M.~O. Ball\egroup{}, Computational complexity of network
  reliability analysis: An overview,  \emph{IEEE Transactions on Reliability}
  \textbf{35} (1986), 230 --239. \doi{10.1109/TR.1986.4335422}.

\bibitem[BP83]{BallProvan1983b}
\bgroup\scshape{}M.~O. Ball\egroup{} and \bgroup\scshape{}J.~S.
  Provan\egroup{}, The complexity of counting cuts and of computing the
  probability that a graph is connected,  \emph{SIAM J. Computing} \textbf{12}
  (1983), 777--788.

\bibitem[CCR{\etalchar{+}}13]{SartorITOR2013a}
\bgroup\scshape{}E.~Canale\egroup{}, \bgroup\scshape{}H.~Cancela\egroup{},
  \bgroup\scshape{}F.~Robledo\egroup{}, \bgroup\scshape{}G.~Rubino\egroup{},
  and \bgroup\scshape{}P.~Sartor\egroup{}, On computing the
  2-diameter-constrained {K}-reliability of networks,  \emph{International
  Transactions in Operational Research} \textbf{20} (2013), 49--58.
  \doi{10.1111/j.1475-3995.2012.00864.x}.  Available at
  \url{http://dx.doi.org/10.1111/j.1475-3995.2012.00864.x}.

\bibitem[CP04]{CP2004}
\bgroup\scshape{}H.~Cancela\egroup{} and \bgroup\scshape{}L.~Petingi\egroup{},
  Reliability of communication networks with delay constraints: computational
  complexity and complete topologies,  \emph{International Journal of
  Mathematics and Mathematical Sciences} \textbf{2004} (2004), 1551--1562.

\bibitem[Col99]{Colbourn99reliabilityissues}
\bgroup\scshape{}C.~J. Colbourn\egroup{}, Reliability issues in
  telecommunications network planning,  in \emph{Telecommunications network
  planning, chapter 9}, Kluwer Academic Publishers, 1999, pp.~135--146.

\bibitem[Coo71]{Cook1971}
\bgroup\scshape{}S.~A. Cook\egroup{}, The complexity of theorem-proving
  procedures,  in \emph{Proceedings of the third annual ACM symposium on Theory
  of computing}, \emph{STOC '71}, ACM, New York, NY, USA, 1971, pp.~151--158.
  \doi{10.1145/800157.805047}.  Available at
  \url{http://doi.acm.org/10.1145/800157.805047}.

\bibitem[GJ79]{Garey:1979:CIG:578533}
\bgroup\scshape{}M.~R. Garey\egroup{} and \bgroup\scshape{}D.~S.
  Johnson\egroup{}, \emph{Computers and Intractability: A Guide to the Theory
  of NP-Completeness}, W. H. Freeman and Company, New York, NY, USA, 1979.

\bibitem[Kar72]{Karp1972}
\bgroup\scshape{}R.~M. Karp\egroup{}, Reducibility among combinatorial
  problems,  in \emph{Complexity of Computer Computations}
  (\bgroup\scshape{}R.~E. Miller\egroup{} and \bgroup\scshape{}J.~W.
  Thatcher\egroup{}, eds.), Plenum Press, 1972, pp.~85--103.

\bibitem[PR01]{PR01}
\bgroup\scshape{}L.~Petingi\egroup{} and
  \bgroup\scshape{}J.~Rodriguez\egroup{}, Reliability of networks with delay
  constraints,  in \emph{Congressus Numerantium}, \textbf{152}, 2001,
  pp.~117--123.

\bibitem[PB83]{provan83}
\bgroup\scshape{}S.~J. Provan\egroup{} and \bgroup\scshape{}M.~O.
  Ball\egroup{}, {The Complexity of Counting Cuts and of Computing the
  Probability that a Graph is Connected},  \emph{SIAM Journal on Computing}
  \textbf{12} (1983), 777--788.

\bibitem[RRS13]{rndm13}
\bgroup\scshape{}F.~Robledo\egroup{}, \bgroup\scshape{}P.~G. Romero\egroup{},
  and \bgroup\scshape{}P.~Sartor\egroup{}, A novel interpolation technique to
  address the {Edge-Reliability} problem,  in \emph{RNDM'13 - 5th International
  Workshop on Reliable Networks Design and Modeling (RNDM'13)}, Almaty,
  Kazakhstan, sep 2013, pp.~77--82.

\bibitem[Ros77]{Rosenthal}
\bgroup\scshape{}A.~Rosenthal\egroup{}, Computing the reliability of complex
  networks,  \emph{SIAM Journal on Applied Mathematics} \textbf{32} (1977),
  384--393. \doi{10.1137/0132031}.  Available at
  \url{http://epubs.siam.org/doi/abs/10.1137/0132031}.

\bibitem[Sar13]{sartor:thesis}
\bgroup\scshape{}P.~Sartor\egroup{}, \emph{{Propriétés et méthodes de calcul de
  la fiabilité diamètre-bornée des réseaux}}, Ph.D. thesis, INRIA/IRISA,
  Universit\'{e} de Rennes I, Rennes, France, december 2013.

\bibitem[Val79]{Valiant1979}
\bgroup\scshape{}L.~Valiant\egroup{}, The complexity of enumeration and
  reliability problems,  \emph{SIAM Journal on Computing} \textbf{8} (1979),
  410--421.

\end{thebibliography}
\bibliographystyle{aomalpha}
\end{document}